\documentclass[12pt]{article}
\usepackage[inline,shortlabels]{enumitem}
\usepackage{float}
\usepackage{mathtools}
\usepackage{graphicx}
\usepackage[round]{natbib}
\usepackage[margin=1in]{geometry}
\usepackage{tikz}
\usepackage[english]{babel}
\usepackage{longtable}
\usepackage{color}
\usepackage{hyperref}
\usepackage{amssymb,amsmath,amsthm}
\usepackage{multirow}
\usepackage[titletoc,title]{appendix}
\usepackage{authblk}
\usepackage{setspace}
\usepackage{dsfont}
\usepackage[OT1]{fontenc}
\usepackage{refcount}
\graphicspath{ {./figs/} }

\usepackage{nameref,zref-xr} 
\zxrsetup{toltxlabel} 

\usepackage{subfig}
\usepackage[inline]{trackchanges}
\addeditor{Ivan}
\addeditor{Nima}
\addeditor{Kara}

\newtheorem{theorem}{Theorem}
\AtEndDocument{\refstepcounter{theorem}\label{finalthm}}
\AtEndDocument{\refstepcounter{equation}\label{finaleq}}

{
  \theoremstyle{definition}
  
}
{
  \theoremstyle{definition}
  
}
{
  \theoremstyle{definition}
  
}

\DeclareMathOperator{\dr}{IF}

\newtheorem{lemma}{Lemma}
\AtEndDocument{\refstepcounter{lemma}\label{finallemma}}

\DeclareMathOperator{\expit}{expit}

\DeclareMathOperator{\logit}{logit} \DeclareMathOperator{\var}{\mathsf{Var}}

\renewcommand{\P}{\mathsf{P}}
\newcommand{\m}{\mathsf{m}}
\newcommand{\p}{\mathsf{p}}
\newcommand{\q}{\mathsf{q}}
\newcommand{\g}{\mathsf{g}}
\newcommand{\ps}{\mathsf{t}}
\newcommand{\e}{\mathsf{e}}
\newcommand{\bb}{\mathsf{b}}
\newcommand{\uu}{\mathsf{u}}
\newcommand{\h}{\mathsf{h}}
\newcommand{\vv}{\mathsf{v}}
\newcommand{\rr}{\mathsf{r}}
\newcommand{\y}{\mathsf{b}}
\newcommand{\s}{\mathsf{c}}

\newcommand{\indep}{\mbox{$\perp\!\!\!\perp$}} 

\newcommand{\dd}{\mathrm{d}}
\newcommand{\Pn}{\mathsf{P}_n}

\newcommand{\thetatmle}{\hat\theta_{\mbox{\scriptsize tmle}}}
\newcommand{\thetaos}{\hat\theta_{\mbox{\scriptsize os}}}
 
\newcommand{\one}{\mathds{1}}
 
\newcommand{\E}{\mathsf{E}}

\renewenvironment{proof}{{\it Proof }}{\qed \\}

\DeclarePairedDelimiterX{\norm}[1]{\lVert}{\rVert}{#1}

\pgfdeclarelayer{background}
\pgfsetlayers{background,main}
\usetikzlibrary{arrows,positioning}
\tikzset{
>=stealth',
punkt/.style={
rectangle,
rounded corners,
draw=black, very thick,
text width=6.5em,
minimum height=2em,
text centered},
pil/.style={
->,
thick,
shorten <=2pt,
shorten >=2pt,}
}
\newcommand{\Vertex}[2]
{\node[minimum width=0.6cm,inner sep=0.05cm] (#2) at (#1) {$\footnotesize#2$};
}
\newcommand{\Vertexr}[2]
{\node[rectangle, draw, minimum width=0.6cm,inner sep=0.05cm] (#2) at (#1) {$\footnotesize#2$};
}
\newcommand{\ArrowR}[3]%
{ \begin{pgfonlayer}{background}
\draw[->,#3] (#1) to[bend right=30] (#2);
\end{pgfonlayer}
}
\newcommand{\ArrowL}[3]%
{ \begin{pgfonlayer}{background}
\draw[->,#3] (#1) to[bend left=45] (#2);
\end{pgfonlayer}
}
\newcommand{\EdgeL}[3]%
{ \begin{pgfonlayer}{background}
\draw[dashed,#3] (#1) to[bend right=-45] (#2);
\end{pgfonlayer}
}

\newcommand{\Arrow}[3]%
{ \begin{pgfonlayer}{background}
\draw[->,#3] (#1) -- +(#2);
\end{pgfonlayer}
}
\title{Efficiently transporting causal (in)direct
  effects to new populations under intermediate confounding and with multiple mediators}

\date{\today}

\author[1]{Kara E. Rudolph}
\author[2]{Iv\'an D\'iaz}

\affil[1]{\small Department of
  Epidemiology, Mailman School of Public Health, Columbia University.}
\affil[2]{\small Division of Biostatistics, Department of Population
  Health Sciences, Weill Cornell Medicine.}

\begin{document}

\maketitle

\begin{abstract}
  The same intervention can produce different effects in different
  sites. Transport mediation estimators can estimate the extent to
  which such differences can be explained by differences in
  compositional factors and the mechanisms by which mediating or
  intermediate variables are produced; however, they are limited to
  consider a single, binary mediator. We propose novel nonparametric
  estimators of transported stochastic (in)direct effects that
  consider multiple, high-dimensional mediators and intermediate
  variables. They are multiply robust, efficient, asymptotically
  normal, and can incorporate data-adaptive estimation of nuisance
  parameters. They can be applied to understand differences in
  treatment effects across sites and/or to predict treatment effects in a target site based on outcome data in source sites.
\end{abstract}

\section{Introduction}
The same intervention can produce different effects in different
populations
\citep[e.g.,][]{orr2003moving,miller2015projected,arnold2018implications}. Different
effects could arise from differences in: i) the distribution of
compositional factors that modify aspects of the intervention's
effectiveness (e.g., gender, age), ii) probability take-up or degree
of adherence to the intervention, iii) the mechanism by which
important mediating or intermediate variables are produced, and/or iv)
the mechanism by which the outcome is produced in different
populations, including differences population- or site-level
contextual variables that are predictive of the outcome
\citep{pearl2018transportability}.

Transportability has been defined by \citet{pearl2018transportability}
as the ``license to transfer causal information learned in experimental
studies to a different environment." Previously, we proposed using the
transport graphs of \citet{pearl2018transportability} coupled with a
transport estimator that predicts effects ``transported" to a target
population as a tool for quantitatively examining the extent to which
differences in effect estimates between sites could be explained by
factors i-iii above \citep{rudolph2017robust,
  rudolph2019transporting}.  In this previous work, we developed an
efficient and robust semi-parametric estimator of transported
stochastic \citep[also called randomized interventional,
see][]{didelez2006direct,vanderweele2014effect} direct and indirect
(what we refer to as (in)direct) effects in a target population
\citep{rudolph2019transporting}. Although this previous estimator
accounted for the presence of intermediate variables (those affected
by treatment/exposure that could affect downstream mediator and
outcome variables), it was limited in that it could only consider
binary versions of a treatment/exposure variable, intermediate
variable, and mediator variable, and assumed that the distribution of
the mediator was known \citep{rudolph2019transporting}. To our
knowledge, it is currently the only available estimator for
transporting (in)direct effects.

However, many research questions involve continuous and/or multiple
mediator variables. Thus, we address this methodologic gap by
proposing novel nonparametric estimators of transported stochastic
(in)direct effects that allow for multiple, possibly high-dimensional
mediators without constraints on their distributions or the
intermediate variables.

To motivate this work, we consider a research question from the Moving
to Opportunity study (MTO), a multi-site randomized controlled trial
conducted by the US Department of Housing and Urban Development, where
families living in high-rise public housing were randomized to receive
a Section 8 housing voucher that they could use to move to a rental on
the private market \citep{sanbonmatsu2011moving}. Families were
followed up at two subsequent time points with the final time point
occurring 10-15 years after randomization. In this study,
some 
unintended harmful effects on children's mental health, substance use,
and risk behavior outcomes were documented
\citep{sanbonmatsu2011moving}, and these overall effects were
partially mediated by aspects of the peer and school environments
\citep{rudolph2018mediation}. However, these unintended harmful
effects and their indirect effect components were not universal across
sites \citep{rudolph2018composition, rudolph2020using}. Our goal is to
use the transportability framework and the novel estimators we propose
to shed light on possible reasons why the intervention had harmful
effects in some sites, particularly in Chicago, but not in others. For
example, if we take Chicago as the site we would like to transport to,
then we borrow information from the remaining sites to learn the
outcome model, we can predict the effect for Chicago, standardizing based on the covariates, intermediate and mediating variables. 

The utility of borrowing or transporting information across sites
applies more broadly than the above MTO example. It applies to
questions that seek to: 1) understand differences in treatment,
policy, or intervention effects across sites in multi-site trials or
cohort studies, or  to 2) predict treatment effects in a target site based on outcome data in source sites.

This paper is organized as follows. In Section 2, we introduce
notation, define the structural causal models we consider, and define
and identify the transported stochastic (in)direct effects. In Section
3, we describe the efficient influence function (EIF), including a
re-parameterization that allows for estimation with multiple and/or
continuously distributed mediators, and derive the robustness
properties of the EIF. In Section 4, we describe two efficient
estimators for the transported stochastic (in)direct effects, based on
the EIF derived in Section 3: an estimator that solves the EIF in one
step and a targeted minimum loss-based estimator (TMLE). In Section 5,
we present results from a simulation study in which we demonstrate the
consistency, efficiency and robustness of the two estimators across
various scenarios. In Section 6, we apply the two estimators to
estimate the transported indirect effects of housing voucher receipt
on subsequent behavioral problems as adolescents among girls in
Chicago, operating through aspects of the school environment,
borrowing information from the other MTO sites. Section 7 concludes
the manuscript.

\section{Notation and definition of (in)direct effects}
Let $O = (S, W, A, Z, M, SY)$ represent the observed data, where $S$
denotes a binary variable indicating membership in the source
population ($S=1$) or target population ($S=0$), $W$ denotes a vector
of observed pre-treatment covariates, $A$ denotes a categorical
treatment variable, $Z$ denotes an intermediate variable (a
mediator-outcome confounder affected by treatment), $M$ denotes a
multivariate mediator, and $Y$ denotes a continuous or binary
outcome. Let $O_1, \ldots, O_n$ denote a sample of $n$
i.i.d.~observations of $O$. Note that the outcome is only observed for
the source population/sites, $S=1$, but we are interested in
estimating effects for the target population/site, $S=0$. We formalize
the definition of our counterfactual variables using the following
non-parametric structural equation model \citep[NPSEM, ][]{Pearl2009}
though equivalent methods may be developed by taking the
counterfactual variables as primitives \citep{Rubin74}. Assume the
data-generating process satisfies:
\begin{multline}\label{eq:npsem}
  S=f_S(U_S);\ W = f_W(S,U_W);\ A = f_A(S,W, U_A);\ Z=f_Z(S,W, A, U_Z);\\
  M = f_M(S,W, A, Z, U_M);\ Y = f_Y(W, A, Z, M, U_Y).
\end{multline}
Here, $U=(U_S,U_W,U_A,U_Z,U_M,U_Y)$ is a vector of exogenous factors, and
the functions $f$ are assumed deterministic but unknown. We use $\P$
to denote the distribution of $O$. 
We let $\P$ be an element of the nonparametric statistical model
defined as all continuous densities on $O$ with respect to some
dominating measure $\nu$. Let $\p$ denote the corresponding
probability density function. We denote random variables with capital
letters and realizations of those variables with lowercase
letters. 
 We define $\P f = \int f(o)\dd \P(o)$ for a given function $f(o)$.

 We use the following additional definitions.  The function $\s(a,z,m,w)$ denotes
 $\P(S=1\mid A=a,Z=z,M=m,W=w)$, $\g(a \mid w)$ denotes
 $\P(A=a\mid W=w,S=0)$, $\e(a \mid m, w)$ denotes
 $\P(A=a\mid M=M,W=w,S=0)$, $\q(z \mid a,w)$ denotes the density of
 $Z$ conditional on $(A,W,S)=(a,w,0)$, $\rr(z \mid a,m,w)$ denotes the
 density of $Z$ conditional on $(A,M,W,S)=(a,m,w,0)$, $\y(a,z,m,w)$
 denotes $\E(Y \mid A = a,Z = z,M = m, W = w, S=1)$, and $\ps$ denotes
 $\P(S=0).$
 For a random variable $X$, we let $X_a$ denote the counterfactual
 outcome observed in a hypothetical world in which $\P(A=a)=1$. For
 example, we have $Z_a = f_Z(S,W,a,U_Z)$, $M_a=f_M(S,W,a,Z_a,U_M)$,
 and $Y_a=f_Y(W,a,Z_a,M_a,U_Y)$. Likewise, we let
 $Y_{a,m} =f_Y(W,a,Z_a,m,U_Y)$ denote the value of the outcome in a
 hypothetical world where $\P(A=a,M=m)=1$.

\subsection{Transported stochastic (in)direct effects}
\label{sec:estimand}
We define the total effect of $A$ on $Y$ in the
target population $S=0$ in terms of a contrast between two user-given
values $a', a^{\star} \in \mathcal A$ among those for whom $S=0$. The total effect can be decomposed into the natural direct
and indirect effect. However, natural direct and indirect effects are not generally identified in the presence of a
mediator-outcome confounder affected by treatment ($Z$, using our notation above)
\citep{avin2005identifiability,tchetgen2014identification}. Direct and indirect effects may be alternatively defined considering a stochastic intervention on the mediator \citep{petersen2006estimation,van2008direct,zheng2012targeted,
  vanderweele2014effect,rudolph2017robust}. 
Let $G_a$ denote a random draw from the conditional distribution of
$M_a$ conditional on $(S,W)$. The stochastic indirect effect (also
called randomized interventional indirect effect) among those for whom
$S=0$ can be written:
$\E(Y_{a', G_{a'}} - Y_{a', G_{a^{\star}}}\mid
S=0)$. 
This is the effect of $A$ on $Y$ that operates through $M$
. The stochastic direct effect among those for whom $S=0$ can be
similarly written:
$\E(Y_{a', G_{a^{\star}}} - Y_{a^{\star}, G_{a^{\star}}}\mid S=0)$,
and is the effect of $A$ on $Y$ that does not operate through $M$.
We focus on identification and estimation
of $\theta = \E(Y_{a', G_{a^{\star}}}\mid S=0)$. Contrasts of $\theta$
under the values of $a'$ and $a^*$ given in the above definitions
correspond to the transported stochastic (in)direct
effects. 
Under the assumptions

\begin{enumerate}[label=(\roman*)]
\item $Y_{a,m}\indep A\mid W$,\label{ass:ncay}
\item $M_{a}\indep A\mid W$,\label{ass:ncam}
\item $Y_{a,m}\indep M\mid (A,W,Z)$,\label{ass:ncmy}
\item
  $\E(Y\mid A = a,Z = z,M = m, W = w, S=1)=\E(Y\mid A = a,Z = z,M = m,
  W = w, S=0)$,\label{ass:exch}
   and 
 \item there is a non-zero probability of assigning any level $A$ for
   all $S,W$; a non-zero probability of assigning any level $A$ for
   all $S=0,W,M$; a non-zero probability of assigning any level of $Z$
   for all combinations of $S=0, W, A, M$; and a non-zero probability
   that $S=0$ for all combinations of $W, A, Z, M$ (referred to as the
   positivity assumption),\label{ass:pos}
\end{enumerate}
$\theta$ is identified and is equal to
\begin{equation}
\theta = \int \y(a',z,m,w)\q(z\mid
  a',w)\p(m\mid a^{\star},w)\p(w\mid S=0)\dd \nu(w,z,m).\label{eq:thetadef}
\end{equation} Assumption \ref{ass:ncay}
states that, conditional on $W$, there is no unmeasured confounding of
the relation between $A$ and $Y$; assumption \ref{ass:ncam} states
that conditional on $W$ there is no unmeasured confounding of the
relation between $A$ and $M$; \ref{ass:ncmy} states that
conditional on $(A,W,Z)$ there is no unmeasured confounding of the
relation between $M$ and $Y$; \ref{ass:exch} states that there is a
common outcome model across populations/ sites. It is this last
assumption \ref{ass:exch} that allows us to transport or borrow
information on the outcome model from other sites. If an alternative
data source is available where $Y$ is observed
among those for whom $S=0$
, then the null hypothesis of equivalence between $S=0$ and
$S=1$ can be tested nonparametrically \citep{luedtke2019omnibus}.
%

\section{Efficient influence function for $\theta$}

The \textit{efficient influence function} (EIF) characterizes the
asymptotic behavior of all regular and efficient estimators
\citep{Bickel97, van2002part}. In addition to being locally efficient,
estimators constructed using the EIF have advantages of multiple
robustness, which means that some components of the data distribution
(i.e., nuisance parameters) can be inconsistently estimated while the
estimator remains consistent. The multiple robustness property also
allows the use data-adaptive machine learning algorithms in estimating
nuisance parameters while retaining the ability to compute correct
standard errors and confidence intervals. This is due to fact that the
asymptotic analysis of the estimators yield second-order bias terms in
differences of the nuisance parameters, and therefore allow slow
convergence rates (e.g., $n^{-1/4}$) for estimating these nuisance
parameters.

\begin{theorem}[Efficient influence function]\label{theo:eif}
  For fixed $a'$, $a^{\star}$ define
\begin{equation}
  \begin{split}
    \h(a, z, m, w) & =  \frac{\p(m\mid a^{\star}, w)}{\p(m\mid a, z,
      w)}\\
    \uu(z,a,w) &=\int_{\mathcal
      M}\y(a,z,m,w)\p(m\mid a^{\star},w)\dd\nu(m)\\
    \vv(a,w)&=\int_{\mathcal
      M\times Z}\y(a',z,m,w)\q(z\mid
    a',w)\p(m\mid a,w)\dd\nu(m,z).
  \end{split}\label{eq:defhuv}
\end{equation}
The efficient influence function for $\theta$ in the
nonparametric model $M$ is equal to
\begin{align}
\label{eq:eif}
\begin{split}
D_{\P,\theta}(o) = & D_{\P,Y}(o) + D_{\P,Z}(o) + D_{\P,M}(o) + D_{\P,W}(o), \text{ where } \\
  D_{\P,Y}(o) = &\frac{\one\{s=1,a=a'\}}{\ps\times \g(a'\mid w)}\frac{1-\s(a,z,m,w)}{\s(a,z,m,w)}\h(a',z,m,w)\{y - \y(a',z,m,w)\}
  \\
  D_{\P,Z}(o) =  & \frac{\one\{s=0,a=a'\}}{\ps\times \g(a'\mid w)}\left\{\uu(z,a',w)-\int_{\mathcal
            Z}\uu(z,a',w)\q(z\mid a',w)\dd\nu(z)\right\}
            \\
  D_{\P,M}(o) =  &  \frac{\one\{s=0,a=a^{\star}\}}{\ps\times \g(a^{\star}\mid w)}\left\{\int_{\mathcal
            Z}\y(a',z,m,w)\q(z\mid a',w)\dd\nu(z)-\vv(a^{\star},
            w)\right\}
            \\
            D_{\P,\theta,W}(o) =  & \frac{\one\{s=0\}}{\ps}\left\{\vv(a^{\star},w) - \theta\right\}
  \end{split}
\end{align}
\end{theorem}

This theorem makes two important contributions that advance the
previous work deriving the EIF for a similar $\theta$, but one that
was limited in that
it 
i) assumed that the distribution of $M$ conditional on $(A,W,S)$ was
known, and ii) could only consider a single binary $M$
\citep{rudolph2019transporting}. First, the EIF we derive does not
assume that that the distribution of $M$ conditional on $(A,W,S)$ is
known, reflected in the $D_{\P,M}(o)$ component of the EIF in Equation
\ref{eq:eif},
above. 
Second, we can overcome the challenge of estimating multivariate or
continuous densities on the mediator, $M$, and intermediate variable,
$Z$, as well as integrals with respect to these densities, if either
$M$ or $Z$ is low-dimensional (though it can be multivariate) by using
an alternative parameterization of the densities that allows
regression methods to be used in estimating the relevant
quantities. In the remainder of this work, we assume $Z$ is
low-dimensional (e.g., binary, as in our MTO illustrative
application), though similar parameterizations may be achieved if $M$
is low-dimensional.

The EIF given in Theorem~\ref{theo:eif} may be represented in terms of
the expressions given in Lemma~\ref{lemma:aeif} below, which does not
depend on conditional densities or integrals on the mediating variables.

\begin{lemma}[Alternative representation of the EIF for univariate $Z$
  and multivariate $M$]\label{lemma:aeif}
  The functions $\h$, $\uu$, and $\vv$ may be parameterized:
  \begin{align}
    \h(a, z, m, w)&=\frac{\g(a\mid w)}{\g(a^{\star}\mid w)}
                    \frac{\q(z\mid a,w)}{\rr(z\mid a,m,w)}
                    \frac{\e(a^{\star}\mid m, w)}{\e(a\mid m,w)}\label{eq:param}\\
    \uu(z,a,w) &= \E\left\{\y(A,Z,M,W)\h(A,Z,M,W),\bigg|\,
                 Z=z,A=a,W=w,S=0\right\},\label{eq:phi}\\
    \vv(a,w) &= \E\left\{\int_{\mathcal
               Z}\y(a',z,M,W)\q(z\mid a',W)\dd\nu(z)\,\bigg|\, A=a,W=w,S=0\right\}.\label{eq:lambda}
  \end{align}
\end{lemma}
In the remainder of the paper, we denote
$\eta=(\s,\g,\e,\q,\rr,\y,\uu,\vv)$ and $D_{\P,\theta}(o) = D_{\eta,\theta}(o)$. We let
$\hat{\eta}$ denote an estimator of $\eta$, and $\eta_1$ denotes the
probability limit of $\hat\eta$, which may be different from the true
value. 
We derive the robustness properties of $D_{\eta,\theta}(O)$ in the
Supplementary Materials; they are given below in
Lemma~\ref{lemma:robust}. 
The behavior of the term  $\P D_{\eta_1,
  \theta}$ determines the
robustness properties of the EIF as an estimating equation. Theorem
1 
in the Supplementary Materials, together with the
Cauchy-Schwarz inequality shows that $\P D_{\eta_1,\theta}$ yields a term of
the order of:
\begin{align*}
  R(\eta_1,\eta) &= \lVert \vv_1 - \vv \lVert \lVert \g_1 -
                   \g\lVert  \\
                 &+\lVert \uu_1 - \uu \lVert \lVert \q_1 -
                   \q\lVert\\
                 &+\lVert \y_1 - \y \lVert \lVert \q_1 -
                   \q\lVert\\
                 &+\lVert \y_1 - \y \lVert \{\lVert \s_1 - \s
                   \lVert + \lVert \q_1 - \q \lVert + \lVert \rr_1
                   - \rr \lVert + \lVert \e_1 - \e \lVert\}
\end{align*}
 such that consistent estimation of $\theta$ is possible
under consistent estimation of certain configurations of the
parameters in $\eta$.  The following lemma is a direct consequence.
\begin{lemma}[Multiple robustness of $D_{\eta,\theta}(O)$]\label{lemma:robust}
  Let $\eta_1=(\s_1,\g_1,\e_1,\q_1,\rr_1,\y_1,\uu_1,\vv_1)$ be such that one
  of the following conditions hold:
  \begin{enumerate}[i)]
  \item $\vv_1=\vv$ and either $(\s_1,\q_1,\e_1,\rr_1)=(\s,\q,\e,\rr)$
    or $(\y_1,\q_1)=(\y,\q)$ or $(\y_1,\uu_1)=(\y,\uu)$, or
  \item $\g_1=\g$ and either $(\s_1,\q_1,\e_1,\rr_1)=(\s,\q,\e,\rr)$
    or $(\y_1,\q_1)=(\y,\q)$ or $(\y_1,\uu_1)=(\y,\uu)$.
  \end{enumerate}
  Then $\P D_{\eta_1,\theta} = 0$ with $D_{\eta,\theta}$ defined as in
  Theorem~\ref{lemma:aeif}.
\end{lemma}
 We note that the cases $(\y_1,\vv_1,\uu_1)=(\y,\vv,\uu)$ and
$(\y_1,\g_1,\uu_1)=(\y,\g,\uu)$ may be uninteresting if the
re-parametrization in Lemma~\ref{lemma:aeif} is used to estimate the
EIF, because in that case, consistent estimation
of $\uu$ and $\vv$ will generally require consistent estimation of
$(\y,\s,\q,\rr,\e)$ in addition to the outer conditional expectations in
Equations (\ref{eq:phi}) and (\ref{eq:lambda}). 

\section{Estimators}
\label{sec:est}
We describe two efficient, robust estimators of $\theta$. In
subsection \ref{sec:os}, we propose an estimator that solves the EIF
estimating equation in one step \citep{pfanzagl1982contributions}
(which we refer to as a one-step estimator), and in subsection
\ref{sec:tmle}, we propose a targeted minimum loss-based estimator
\citep[TMLE, ][]{vdl2006targeted}, which is a substitution estimator
that also solves the EIF estimating equation, but does it through
iterative de-biasing targeted updates to nuisance parameters. We
provide the R code to implement the proposed estimators, freely
available at \url{https://github.com/kararudolph/transport}.

Let $\thetaos$ and $\thetatmle$ denote the estimators defined below in Sections~\ref{sec:os} and \ref{sec:tmle}. Per the theorem below, the two estimators are asymptotically normal and efficient. 
\begin{theorem}[Asymptotic normality and efficiency]\label{theo:assymp}
  Assume
  \begin{enumerate}[label=(\roman*)]
  \item Positivity, described as identification assumption
    \ref{ass:pos} in Section
    \ref{sec:estimand}, \label{ass:positivity} and
  \item The class of functions
    $\{D_{\eta,\theta}:|\theta-\theta_0| <\delta, ||\eta-\eta_1||<\delta\}$ is
    Donsker for some $\delta >0$ and such that
    $P_0(D_{\eta,\theta} - D_{\eta_1,\theta_0})^2\to 0$ as
    $(\eta,\theta)\to (\eta_1,\theta_0)$, and
  \item The second-order term $R(\hat\eta,\eta)$ is
    $o_P(n^{-1/2})$. \label{ass:remainder}
  \end{enumerate}
  Then, $\sqrt{n}(\thetaos-\theta) \rightarrow N(0,\sigma^2)$, and $\sqrt{n}(\thetatmle-\theta) \rightarrow N(0,\sigma^2)$ where
  $\sigma^2=\var(D_{\eta}(O))$ is the non-parametric efficiency bound.
\end{theorem}

The proof of this theorem follows the general proof presented in Appendix 18 of
\cite{vanderLaanRose11}.  As a consequence, the
variance of the estimators that follow can be estimated as the sample
variance of the EIF, with $\hat{\theta}$ and the nuisance parameters
estimated as described above. This variance estimate may be used to construct Wald-type confidence intervals.

The Donsker condition of Theorem \ref{theo:assymp} may be
avoided by using cross-fitting
\citep{klaassen1987consistent,zheng2011cross, chernozhukov2016double}
in the estimation procedure. Let ${\cal V}_1, \ldots, {\cal V}_J$
denote a random partition of the index set $\{1, \ldots, n\}$ into $J$
prediction sets of approximately the same size. That is,
${\cal V}_j\subset \{1, \ldots, n\}$;
$\bigcup_{j=1}^J {\cal V}_j = \{1, \ldots, n\}$; and
${\cal V}_j\cap {\cal V}_{j'} = \emptyset$. In addition, for each $j$,
the associated training sample is given by
${\cal T}_j = \{1, \ldots, n\} \setminus {\cal V}_j$. let
$\hat \eta_{j}$ denote the estimator of $\eta$, obtained by training
the corresponding prediction algorithm using only data in the sample
${\cal T}_j$. Further, we let $j(i)$ denote the index of the
validation set which contains observation $i$. The one-step and TMLE
estimators may be adapted to cross-fitting by substituting all
occurrences of $\hat\eta(O_i)$ by $\hat\eta_{j(i)}(O_i)$ in the
respective algorithms.

The third condition of Theorem \ref{theo:assymp} can be satisfied by many data adaptive algorithms (e.g., lasso \citep{bickel2009simultaneous}, regression tress \citep{wager2015adaptive}, neural networks \citep{chen1999improved}, highly adaptive lasso (HAL)\citep{van2017generally}); we use HAL in the simulations that follow. 

\subsection{One-step estimator}\label{sec:os}
The one-step estimate of $\theta$ is given by the solution to the EIF
estimating equation:


\[\thetaos=\frac{1}{n}\sum_{i=1}^n \{D_{\hat\eta, Y}(O_i) + D_{\hat\eta, Z}(O_i) + 
  D_{\hat\eta, M}(O_i)\} + \frac{1}{n}\sum_{i=1}^n
  \frac{\one\{S_i=0\}}{\hat \ps}\hat\vv(a^*, W_i).\]

We first describe how to estimate $D_{\eta,Y}$. The regression
$\y(a^{\prime},z,m,w)$ can be estimated by fitting a regression of $Y$
on $W, A, Z, M$ among observations with $S=1$ and then predicting
values of $Y$ setting $A=a^\prime$. The probability $\ps$ is estimated
as the empirical proportion of observations with $S=0$ (i.e., in the
target population). The regression function $\s(a^{\prime},z,m,w)$ can
be estimated by fitting a regression of $S$ on $W,A,Z,M$ and
predicting the probability that $S=1$ setting $A=a'$. The treatment
mechanism $\g(a\mid w)$ for $a\in\{a',a^*\}$ can be estimated by
fitting a regression of $A$ on $(S,W)$ and predicting the probability
that $A=a$, setting $S=0$. For the motivating example we consider here
in which assignment of $A$ is randomized, these can be estimated as
the empirical probabilities that $A=a'$ and $A=a^*$ among those with
$S=0$. Under the reparameterization in Lemma \ref{lemma:aeif} and in our motivating example, $\q(z\mid a,w)$ can be estimated by
fitting a regression of $Z$ on $S,A,W$ and predicting the probability
that $Z=z$ setting $A=a^{\prime}, S=0$. Likewise, $\rr(z\mid a, m, w)$
can be estimated by fitting a regression of $Z$ on $S,A,M,W$ and
predicting the probability that $Z=z$ setting $A=a^{\prime}, S=0$. The
treatment probabilities $\e(a^{\prime}\mid m,w)$ and
$\e(a^* \mid m,w)$ can be estimated by fitting a regression of $A$ on
$S,M,W$ and predicting the probability that $A=a'$ and $A=a^*$,
respectively, setting $S=0$.

We next describe how to estimate $D_{\eta,Z}$. For binary $Z$, the EIF
simplifies to be 
\[D_{\eta,Z}(o) = \frac{\one\{s=0,a=a'\}}{\ps\times \g(a'\mid
    w)}\{\uu(1,a',w)-\uu(0,a',w)\}\left\{z - \q(1\mid a',w)\right\}.\] The parameters
$\ps, \g(a^{\prime}\mid w)$ and $\q(z \mid a^{\prime},w)$ can be
estimated as described above. For each $z$,
$\uu(z,a^{\prime},w)$ can be estimated by regressing the quantity
$\y(A,Z,M,W) \times
\h(A,Z,M,W)$ 
on $S,A,Z,W$ and getting predicted values, setting $Z=z, A=a^{\prime},
S=0$.

To estimate $D_{\eta,M}$, we estimate $\ps, \g(a^*\mid w),
\y(a^{\prime}, z,m,w), \q(z\mid a^{\prime},
w)$ as described above. The function
$\vv(a^*,w)$ can be estimated by marginalizing out
$Z$ from $\y(a^{\prime},z,M,W)$ using $\q(z\mid
a^{\prime},W)$ as predicted probabilities for each $z$, and then regressing the resulting quantity on
$A,W,S,$ and predicting values setting $A=a^*, S=0$.


\subsection{TML estimator}
\label{sec:tmle}
We now describe how to compute a related TML estimator. We
assume $Y$ can be bounded in $[0,1]$, as described previously \citep{gruber2010targeted}. Many of the steps are identical to those for
the one-step estimator, the differences are in the targeting of
$ \y(a^{\prime},z,m,w)$, $\q(z\mid a^{\prime},w)$, and
$\vv(a^*,w)$. 

Let $\hat{\y}(a^{\prime},z,m,w)$ be an initial estimate of
$\y(a^{\prime},z,m,w)$. We update this initial estimate using
covariate
\[\hat C_\y(A,Z,M,W) =
  \frac{1-\hat \s(A,Z,M,W)}{\hat \s(A,Z,M,W)}\frac{\hat \h(A,Z,M,W)}{\hat \g(A\mid
    W)\times \hat\ps}\] in a logistic regression of $Y$ with
$\logit \hat{\y}(A,Z,M,W)$ as an offset, among the subset
for which $A=a^{\prime}, S=1$. Let $\hat{\epsilon}_\y$ denote the MLE
fitted coefficient associated with $\hat C_\y(A,Z,M,W)$. The targeted (i.e.,
updated) estimate 
is given by
\[\logit\tilde{\y}(a',z,m,w) = \logit
  \hat{\y}(a^{\prime},z,m,w) + \hat{\epsilon}_\y \hat C_\y(a',z,m,w).\]
An alternative algorithm would use 
\[\frac{\hat \e(a^*|M,W)}{\hat
    \e(A\mid M,W)}\frac{1}{\hat\g(a^\star\mid W) \times \hat \ps}\] as weights of what would become a
weighted logistic regression model with covariate
\[\hat C_\y(A,Z,M,W) = \frac{1-\hat \s(A,Z,M,W)}{\hat \s(A,Z,M,W)}\frac{\hat \q(Z\mid A,W)}{\hat \rr(Z\mid A,M,W)}.\]
Next, let $\hat{\q}(z\mid a^{\prime},w)$ be an initial estimate of
$\q(z\mid a^{\prime},w)$. We update this initial estimate using
covariate
\[\hat C_\q(A,W) = \frac{\hat \uu(A,1,W) - \hat \uu(A,0,W)
  }{\hat\g(A\mid W)\times \hat\ps}\] in a logistic regression of $Z$
with $\logit \hat{\q}(A \mid A,W))$ as an offset, among the subset for
which $A=a^{\prime}, S=0$. Let $\hat{\epsilon}_\q$ be the MLE fitted
coefficient associated with $C_\q(A,W,S)$. The targeted estimate is
given by
\[\logit\tilde{\q}(z\mid a', w) = \logit
  \hat{\q}(z\mid a^{\prime}, w) + \hat{\epsilon}_\q \hat C_\q(a',w).\]
To potentially improve performance in finite samples, we can move
$\{\hat \g(A\mid W)\times \hat\ps\}^{-1}$ into the weights of a
weighted logistic regression model, leaving
$\hat\uu(A,1,W) - \hat\uu(A,0,W)$ as
$\hat C_\q(A,W)$. 
 
Replacing $\hat{\y}$ and $\hat{\q}$ with $\tilde{\y}$ and
$\tilde{\q}$, the above steps can be iterated until the score equation
$n^{-1}\sum_i\{D_{\tilde{\eta}, Y}(O_i) + D_{\tilde{\eta},
  Z}(O_i)\}=0$ is solved up to a factor of
$(\sqrt{n}\log(n))^{-1}$. This iterating process and stopping
criterion ensures that the efficient influence function is solved up
to $n^{-1/2}$ 
and mitigates risk of overfitting. 

Next, we marginalize out $Z$ from $\Tilde{\y}(a^{\prime},z,M,W)$ using
$\Tilde{\q}(z\mid a^{\prime},W)$ as predicted probabilities for each
$z$, and call the resulting quantity $Q$. This
quantity is then regressed on $(A, W)$ among units with $S=0$ to
obtain an estimator $\hat{\vv}(A, W)$.  
This estimate is updated using covariate
\[\hat C_\vv(A,W) = \frac{1}{\hat \g(A\mid W)\times \hat \ps}\] in a
logistic regression of $Q$ with $\logit \hat{\vv}(A,W))$ as an
offset, among the subset for which $A=a^*, S=0$. Let
$\hat{\epsilon}_\vv$ denote the MLE fitted coefficient on
$C_\vv(A,W,S)$. The targeted estimate is given by
\[\logit\tilde{\vv}(a^*, w) = \logit
  \hat{\vv}(a^*, w) + \hat{\epsilon}_\vv \hat C_\vv(a^*,w).\] To
potentially improve finite sample performance, $\hat C_\vv(A,W)$ may
be moved into the weights of a weighted logistic regression model with
intercept
only. 
The empirical mean of $\tilde{\vv}(a^*, W_i)$ among those for whom
$S=0$ is the TMLE estimate. Its variance can be estimated as
the sample variance of the estimated EIF, given in Eq \ref{eq:eif}.

\section{Simulation}
\label{sec:sim}
We conducted a limited simulation study to examine and compare finite
sample performance of these two estimators. We consider the
data-generating mechanism (DGM) as follows. All variables are
Bernoulli distributed with probabilities given by
{\footnotesize
  \begin{align*}
    P(W_1 = 1) &= 0.5                                                                                     \\
    P(W_2 = 1 \mid W_1) &= 0.4 + 0.2W_1                                                                   \\
    P(\Delta = 1 \mid W) &= \expit(-1 + \log(4)W_1 + \log(4)W_2)                                                                                      \\
    P(S = 1 \mid \Delta, W) &= \expit(\log(1.2)W_1 + \log(1.2)W_2 + \log(1.2)W_1W_2)                                \\
    P(A = 1\mid S, \Delta, W) &= 0.5                                                                                        \\
    P(Z = 1 \mid A, S, \Delta, W) &= \expit(-\log(2) + \log(4)A + -\log(2)W_2 + \log(1.4)S + \log(1.43)A\times S ) \\
    P(M = 1 \mid Z, A, S, \Delta, W) &= \expit(-\log(2) + \log(4)Z - \log(1.4)W_2 + \log(1.4)S)                        \\
    P(Y = 1 \mid M, Z, A, S, \Delta, W) &= \expit(-\log(5) + \log(8)Z + \log(4)M - \log(1.2)W_2 + \log(1.2)W_2Z          
  \end{align*}}

This DGM is formulated to align with features of the MTO study we use for
the illustrative example. For example, $A$ is randomly assigned and
adheres to the exclusion restriction
\citep{angrist1996identification}, aligned with its role as an
instrumental variable. In addition, we consider a modification of the
observed data we have considered thus far:
$\Delta\times O=\Delta\times (S, W, A, Z, M, S Y)$, where $\Delta$ is
an indicator of selection into the survey sample. We assume the survey
sampling weights are known or can be estimated as
\[\hat\Gamma_i = \frac{1}{\Pi_i}\frac{\sum_{i=1}^n
    (1-S_i)}{\sum_{i=1}^n (1-S_i)\Pi_i^{-1}},\] where
$\Pi = \P(\Delta = 1 \mid X)$ and $X$ represents unobserved variables
used in the sampling design. Our previous identification result, which
can alternatively be written as $\theta=\E[\vv(a^\star, W)\mid S=0]$,
then becomes
\begin{align*}
  \theta_{\Delta=1} & = \E\left[\Gamma\,\,\vv(a^\star, W)\,\bigg|\, S=0,\Delta=1  \right],  
\end{align*}
where we have added an index $\Delta=1$ to emphasize that we are interested in parameters for the population from which the sample was drawn. 
The EIF is modified to be $D_{\P,\Delta=1}(o)=\Gamma D_{\P}(o)$, and the estimators of the previous section can be applied by using the weights $\hat\Gamma_i$ for each subject in the sample.  

We consider estimator performance in terms of absolute bias, absolute bias scaled by $\sqrt{n}$, influence curve-based standard error relative to the Monte Carlo-based standard error, standard deviation of the estimator relative to the efficiency bound scaled by $\sqrt{n}$, mean squared error relative to the efficiency bound scaled by $n$, and 95\% confidence interval (CI) coverage. We run 1,000 simulations for sample sizes N=1,000 and N=10,000. We also consider several model specifications. One in which all nuisance parameters in $\eta$ are correctly specified, others that misspecify each nuisance parameter one at a time, another in which $\g(a^{\prime} \mid w), \y(a^{\prime},z,m,w), \q(z \mid a^{\prime},w)$ are correctly specified but the rest are not; and last, correctly specifying $\y(a^{\prime},z,m,w), \q(z \mid a^{\prime},w),\vv(a^*,w)$ but incorrectly specifying the rest. 
Under correct specification scenarios, we use HAL\citep{benkeser2016highly,van2017generally} to fit each nuisance parameter. For incorrect specification, we use an intercept-only model. 

Table \ref{tab:simdirect} shows simulation results for the transported stochastic direct effect, and Table \ref{tab:simindirect} shows simulation results for the transported stochastic indirect effect comparing the one-step and TML estimators under correct specification of all nuisance parameters and various misspecifications. Given the robustness results in Lemma \ref{lemma:robust}, we expect consistent estimates for all specifications in Tables \ref{tab:simdirect} and \ref{tab:simindirect} except when $\q$ is misspecified. We see this reflected in the results. We see that when the $\q$ model is misspecified, bias is more than an order of magnitude greater than any other specification for the transported stochastic direct effect in Table \ref{tab:simdirect}, and also greater, though to a lesser extent for the transported stochastic indirect effect in Table \ref{tab:simindirect}. 95\% CI coverage using IC-based inference is close to 95\% in the correctly specified scenario, but is poor when $\q$ is misspecified for the transported stochastic direct effect (Table \ref{tab:simdirect}), which is not unexpected given the biased estimates in this scenario. Coverage is less than 95\% in other misspecified scenarios for both the transported direct and indirect effects (e.g., 68\% when the $\y$ model is misspecified for the transported stochastic indirect effect, Table \ref{tab:simindirect}). This is not unexpected; the IC may not provide accurate inference when the IC at the estimated distribution using misspecified models does not converge to the IC at the true distribution. For robustness to extend to IC-based inference, further targeting of the nuisance parameters would be necessary that would preserve asymptotic linearity with a known influence curve at the cost of some efficiency.\citep{van2014targeted,benkeser2016doubly} Lastly, we note that under the smaller sample size of N=1,000 we see some deterioration in performance, particularly for the indirect effect, which is expected given that the true indirect effect is over five times smaller than the direct effect.

\begin{table}
\footnotesize
\centering
\caption{Simulation results for the transported stochastic direct effect.}
\begin{tabular}{|p{3.2cm} | p{1.5cm}| p{1cm} p{1cm} p{1cm} p{1cm}   p{1cm} p{2.1cm}   | }
\hline
Nuisance Parameters Misspecified & Estimator & $|\text{bias}|$& $\sqrt{n}|\text{bias}|$ & relse & relsd & relrmse &95\%CI Cov \\ 
  \hline
  \multicolumn{8}{|l|}{Transported stochastic direct effect}\\ \hline
  \multicolumn{8}{|l|}{N=10,000}\\ \hline
None & os & 0.0005 & 0.0490 & 1.0200 & 0.9489 & 0.9488 & 0.9570 \\ 
& tmle & 0.0004 & 0.0415 & 1.0040 & 0.9610 & 0.9608 & 0.9530 \\ 
$\s$ & os & 0.0005 & 0.0519 & 1.0023 & 0.8226 & 0.8227 & 0.9570 \\ 
 & tmle & 0.0003 & 0.0312 & 0.9577 & 0.8579 & 0.8576 & 0.9460 \\ 
$\g$& os & 0.0005 & 0.0480 & 1.0213 & 0.9481 & 0.9480 & 0.9580 \\ 
 & tmle & 0.0004 & 0.0408 & 1.0055 & 0.9600 & 0.9598 & 0.9520 \\ 
$\e$   & os & 0.0002 & 0.0156 & 1.0097 & 0.9431 & 0.9427 & 0.9520 \\ 
 & tmle & 0.0003 & 0.0301 & 0.9878 & 0.9602 & 0.9599 & 0.9460 \\ 
$\q$  & os & 0.0885 & 8.8488 & 0.7750 & 1.4727 & 5.2339 & 0.0250 \\ 
 & tmle & 0.0348 & 3.4814 & 1.0656 & 1.0154 & 2.2215 & 0.5580 \\ 
 $\rr$  & os & 0.0024 & 0.2382 & 1.0889 & 0.8724 & 0.8824 & 0.9640 \\ 
 & tmle & 0.0021 & 0.2134 & 1.0809 & 0.8788 & 0.8867 & 0.9640 \\ 
 $\y$  & os & 0.0047 & 0.4739 & 1.0460 & 0.9615 & 0.9979 & 0.9470 \\ 
 & tmle & 0.0107 & 1.0661 & 0.9908 & 1.0007 & 1.1690 & 0.9070 \\ 
 $\uu$  & os & 0.0053 & 0.5285 & 0.9400 & 0.9405 & 0.9867 & 0.9230 \\ 
& tmle & 0.0053 & 0.5262 & 0.9249 & 0.9530 & 0.9983 & 0.9150 \\ 
$\vv$   & os & 0.0005 & 0.0499 & 1.0213 & 0.9476 & 0.9476 & 0.9570 \\ 
& tmle & 0.0004 & 0.0421 & 1.0028 & 0.9621 & 0.9619 & 0.9520 \\ 
$\s, \e, \rr, \uu, \vv$  & os & 0.0023 & 0.2293 & 0.8924 & 0.7159 & 0.7272 & 0.9140 \\ 
 & tmle & 0.0019 & 0.1889 & 0.8519 & 0.7465 & 0.7538 & 0.9020 \\ 
$\s, \g, \e, \rr, \uu$& os & 0.0023 & 0.2321 & 0.8914 & 0.7165 & 0.7281 & 0.9140 \\ 
& tmle & 0.0019 & 0.1904 & 0.8548 & 0.7438 & 0.7513 & 0.9030 \\ 
  \hline
  \multicolumn{8}{|l|}{N=1,000}\\ 
   \hline
None&  os & 0.0014 & 0.0454 & 1.0200 & 0.8925 & 0.8921 & 0.9591\\ 
 & tmle & 0.0028 & 0.0880 & 0.9702 & 0.9309 & 0.9314 & 0.9414 \\ 
$\s$ & os & 0.0003 & 0.0104 & 1.0340 & 0.7691 & 0.7683 & 0.9600  \\ 
 & tmle & 0.0021 & 0.0648 & 0.9648 & 0.8190 & 0.8190 & 0.9460\\ 
$\g$ & os & 0.0019 & 0.0617 & 1.0167 & 0.8958 & 0.8957 & 0.9520 \\ 
 & tmle & 0.0032 & 0.1009 & 0.9697 & 0.9317 & 0.9327 & 0.9424 \\ 
 $\e$  & os & 0.0036 & 0.1134 & 1.0131 & 0.8855 & 0.8871 & 0.9520 \\ 
& tmle &0.0049 & 0.1550 & 0.9611 & 0.9252 & 0.9286 & 0.9440\\ 
$\q$ & os & 0.0672 & 2.1252 & 0.7993 & 1.3098 & 1.7796 & 0.7560 \\ 
 & tmle & 0.0280 & 0.8862 & 1.0124 & 0.9771 & 1.0981 & 0.9300 \\ 
$\rr$ & os & 0.0073 & 0.2303 & 1.1242 & 0.8149 & 0.8245 & 0.9620 \\ 
 & tmle & 0.0070 & 0.2202 & 1.0994 & 0.8315 & 0.8400 & 0.9620 \\
$\y$  & os & 0.0047 & 0.1499 & 1.0460 & 0.3041 & 0.3156 & 0.9470 \\ 
& tmle & 0.0107 & 0.3371 & 0.9908 & 0.3164 & 0.3697 & 0.9070 \\ 
$\uu$ & os & 0.0106 & 0.3362 & 0.9735 & 0.8614 & 0.8814 & 0.9420 \\ 
& tmle & 0.0101 & 0.3186 & 0.9234 & 0.8993 & 0.9164 & 0.9180 \\ 
$\vv$ & os & 0.0009 & 0.0295 & 1.0304 & 0.8827 & 0.8819 & 0.9589 \\ 
 & tmle & 0.0021 & 0.0668 & 0.9857 & 0.9152 & 0.9150 & 0.9498 \\ 
$\s, \e, \rr, \uu, \vv$ & os & 0.0030 & 0.0949 & 0.9553 & 0.6643 & 0.6657 & 0.9315 \\ 
& tmle & 0.0013 & 0.0424 & 0.9141 & 0.6898 & 0.6895 & 0.9224 \\ 
$\s, \g, \e, \rr, \uu$ & os & 0.0019 & 0.0591 & 0.9533 & 0.6663 & 0.6664 & 0.9291 \\ 
& tmle & 0.0001 & 0.0034 & 0.9044 & 0.6981 & 0.6974 & 0.9222 \\ 
   \hline
\end{tabular}
\label{tab:simdirect}
\end{table}

\begin{table}
\centering
\footnotesize
\caption{Simulation results for the transported stochastic indirect effect.}

\begin{tabular}{|p{3.2cm} p{1.5cm}| p{1cm} p{1cm} p{1cm} p{1cm}   p{1cm} p{2.1cm}   | }
  \hline
  \label{tab:simindirect}
Nuisance Parameters Misspecified & Estimator & $|\text{bias}|$& $\sqrt{n}|\text{bias}|$ & relse & relsd & relrmse &95\%CI Cov \\ 
  \hline
   \multicolumn{8}{|l|}{Transported stochastic indirect effect}\\
   \hline
  \multicolumn{8}{|l|}{N=10,000}\\ \hline
None & os & 0.0000 & 0.0030 & 0.9966 & 0.9760 & 0.9755 & 0.9420 \\ 
 & tmle & 0.0001 & 0.0065 & 0.9895 & 0.9778 & 0.9774 & 0.9400 \\ 
 $\s$  & os & 0.0003 & 0.0272 & 0.9864 & 0.9445 & 0.9456 & 0.9410 \\ 
 & tmle & 0.0001 & 0.0147 & 0.9734 & 0.9530 & 0.9529 & 0.9370 \\ 
$\g$ & os & 0.0000 & 0.0004 & 0.9976 & 0.9749 & 0.9744 & 0.9430 \\ 
 & tmle & 0.0000 & 0.0044 & 0.9907 & 0.9768 & 0.9764 & 0.9430 \\ 
$\e$  & os & 0.0006 & 0.0632 & 0.9610 & 0.8917 & 0.9003 & 0.9390 \\ 
& tmle & 0.0007 & 0.0668 & 0.9603 & 0.8887 & 0.8983 & 0.9380 \\
$\q$  & os & 0.0020 & 0.2035 & 0.9285 & 1.0709 & 1.1459 & 0.9070 \\ 
 & tmle & 0.0030 & 0.2951 & 0.8061 & 1.2396 & 1.3737 & 0.8410 \\ 
$\rr$  & os & 0.0003 & 0.0324 & 1.0081 & 1.0034 & 1.0051 & 0.9500 \\ 
 & tmle & 0.0001 & 0.0075 & 1.0429 & 0.9652 & 0.9648 & 0.9590 \\ 
$\y$ & os & 0.0001 & 0.0067 & 0.4870 & 1.0879 & 1.0874 & 0.6850 \\ 
& tmle & 0.0001 & 0.0099 & 0.4747 & 1.1763 & 1.1759 & 0.6700 \\ 
$\uu$ & os & 0.0000 & 0.0020 & 0.9997 & 0.9714 & 0.9709 & 0.9440 \\ 
 & tmle & 0.0000 & 0.0009 & 0.9919 & 0.9744 & 0.9739 & 0.9440 \\ 
$\vv$ & os & 0.0000 & 0.0026 & 1.0250 & 1.0114 & 1.0109 & 0.9550 \\ 
& tmle & 0.0001 & 0.0058 & 1.0015 & 1.0259 & 1.0254 & 0.9480 \\ 
$\s, \e, \rr, \uu, \vv$  & os & 0.0006 & 0.0618 & 0.9375 & 0.9834 & 0.9907 & 0.9400 \\ 
& tmle & 0.0007 & 0.0696 & 0.9168 & 0.9947 & 1.0040 & 0.9310 \\ 
$\s, \g, \e, \rr, \uu$ & os & 0.0007 & 0.0676 & 0.9032 & 0.9399 & 0.9492 & 0.9150 \\ 
 & tmle & 0.0008 & 0.0767 & 0.8996 & 0.9387 & 0.9508 & 0.9150 \\ 
\hline
  \multicolumn{8}{|l|}{N=1,000}\\ \hline
None& os &  0.0007 & 0.0229 & 0.9010 & 1.0200 & 1.0201 & 0.9041 \\ 
& tmle & 0.0006 & 0.0200 & 0.8900 & 1.0209 & 1.0207 & 0.8988 \\
$\s$ &os & 0.0013 & 0.0407 & 0.8891 & 0.9901 & 0.9924 & 0.8900 \\ 
 & tmle & 0.0011 & 0.0337 & 0.8729 & 0.9970 & 0.9983 & 0.8880 \\
$\g$  & os & 0.0009 & 0.0293 & 0.9092 & 1.0185 & 1.0194 & 0.9072 \\ 
& tmle & 0.0008 & 0.0242 & 0.8974 & 1.0193 & 1.0197 & 0.8992 \\ 
 $\e$ & os & 0.0026 & 0.0818 & 0.9025 & 0.9447 & 0.9582 & 0.8992 \\ 
& tmle & 0.0025 & 0.0797 & 0.8953 & 0.9415 & 0.9543 & 0.8976 \\ 
  $\q$ & os &0.0017 & 0.0541 & 0.8405 & 1.0920 & 1.0963 & 0.8700 \\ 
& tmle & 0.0017 & 0.0525 & 0.7671 & 1.1978 & 1.2012 & 0.8440 \\ 
$\rr$ & os & 0.0004 & 0.0120 & 0.8955 & 1.0476 & 1.0468 & 0.9080 \\ 
& tmle & 0.0008 & 0.0263 & 0.9121 & 1.0109 & 1.0112 & 0.8980 \\ 
$\y$ & os & 0.0001 & 0.0021 & 0.4870 & 0.3440 & 0.3439 & 0.6850 \\ 
& tmle & 0.0001 & 0.0031 & 0.4747 & 0.3720 & 0.3719 & 0.6700 \\ 
$\uu$ & os & 0.0008 & 0.0264 & 0.8872 & 1.0328 & 1.0331 & 0.8900 \\ 
& tmle & 0.0007 & 0.0215 & 0.8655 & 1.0424 & 1.0423 & 0.8800 \\ 
  $\vv$  & os & 0.0004 & 0.0114 & 0.9605 & 1.0274 & 1.0265 & 0.9247 \\ 
& tmle & 0.0003 & 0.0086 & 0.9279 & 1.0450 & 1.0440 & 0.9110 \\ 
  $\s, \e, \rr, \uu, \vv$  & os & 0.0023 & 0.0733 & 0.8522 & 1.0237 & 1.0331 & 0.8973 \\ 
& tmle & 0.0026 & 0.0831 & 0.8284 & 1.0286 & 1.0410 & 0.8881 \\ 
$\s, \g, \e, \rr, \uu$ & os & 0.0020 & 0.0618 & 0.8595 & 0.9285 & 0.9358 & 0.8741 \\ 
& tmle & 0.0022 & 0.0708 & 0.8492 & 0.9229 & 0.9328 & 0.8741 \\ \hline
\end{tabular}
\end{table}

\section{Illustrative Example}
We apply the one-step and TML estimators proposed in Section
\ref{sec:est} to estimate stochastic indirect effects transported
across MTO sites, as described in the Introduction. Specifically, we
are interested in the extent to which differences in: a) the
distribution of individual-level compositional factors between the
sites, b) take-up of the intervention (i.e., using the housing voucher
to move), and c) distribution of school environment mediating
variables can explain the difference in the indirect effect estimates
between MTO sites.

For this example, we consider the indirect effect of randomized
receipt of a Section 8 housing voucher ($A$) and subsequent use ($Z$)
on behavioral problems ($Y$) \citep{zill1990behavior} through aspects
of the school environment ($M$, i) rank of the schools attended, and
ii) whether ever attended a school in the top 50\% of rankings, iii)
number of schools attended, iv) number of moves since baseline, v)
average proportion of students receiving free or reduced lunch, vi)
ratio of students to teachers, vii) proportion of schools attended
that were Title I, and viii) whether or not the most recent school
attended was in the same district as the baseline school) among girls,
comparing the Los Angeles (LA) and New York City (NYC) sites ($S=1$,
N=1,000) to the Chicago site ($S=0$, N=600). We do this in order to
illustrate our methods: the outcomes in Chicago were actually
observed, so we can compared the transported estimate with estimates
obtained using Chicago outcome data. Variables $W$ and $A$ were measured at baseline, when
the children were 0-10 years old. Mediating variables were measured
during the interval between baseline and the final follow-up timepoint
10-15 years later. The outcome was measured at the final follow-up
timepoint.  We account for a large number of covariates at the child
and family levels: child age, race/ethnicity, history of behavioral
problems, and gifted/talented status; parental education, marital
status, whether or not the parent was under 18 at the birth of the
child, employment, receipt of other public benefits, household size,
feeling like the neighborhood was unsafe at night, feeling very
dissatisfied with the neighborhood, whether or not the family had
previously moved more than three times, wanting to move for better
schools, whether or not the family had received a Section 8 voucher
before, and poverty level of the baseline neighborhood. For this
research question, randomization to receive a Section 8 housing
voucher is an instrumental variable that affects $M$ and $Y$ through
the intermediate variable of using the voucher to move out of public
housing and into a rental on the private market ($Z$). We use the MTO
sampling weights as described in Section \ref{sec:sim}. These weights
account for sampling of children within families, changing
randomization ratios, and loss to
follow-up \citep{sanbonmatsu2011moving}. We use data-adaptive methods
for fitting the nuisance parameters, using a cross-validated ensemble
of machine learning algorithms \citep{van2007super}, that includes
generalized linear models, intercept-only models, and
lasso \citep{tibshirani1996regression} that included all first and
second-order predictors. To estimate the observed, non-transported
stochastic indirect effects, we use non-transported versions of the
one-step and TML estimators developed previously \citep{diaz2019non}.
Standard errors are estimated using the sample variance of the
influence curve.

Figure \ref{fig:examp} shows the transported and observed indirect effect estimates and their 95\% CIs. Looking at the observed estimates, the indirect pathway from housing voucher receipt and use through the school environment to behavioral problems is protective for girls in LA and NYC, resulting in a reduction in behavioral problems at the final time point. However, the same pathway appears harmful for girls in Chicago, resulting in an increase of behavioral problems. Comparing the transported stochastic indirect effect estimate  (one-step estimator: 0.0043, 95\% CI: -0.0150, 0.0237, risk difference scale; TMLE: 0.0153, 95\% CI: -0.0150, 0.0420) to the observed estimate for girls in Chicago (0.0089, 95\% CI: 0.0007, 0.0171), we see that the two are similar even though the outcome data from Chicago was not used in the transported estimates. Thus, by taking the outcome model for LA and NYC and standardizing based on $W, A, Z, M$ in Chicago, the predicted effect for Chicago is close to the observed. In contrast if they were not close to each other, this would suggest that the identification assumptions were not met. In the context of MTO, identification assumption (iv) of a common outcome model is arguably the most tenuous. This assumption would not hold in the presence of any contextual-level effects on the outcome model, such as the local economy, housing market conditions, segregation, etc. 

\begin{figure}
  \caption{Stochastic indirect effects estimates of being randomized
    to the Section 8 voucher group on behavioral problems score in
    adolescence, 10-15 years later, mediated through features of the
    school environment, among girls. Estimates and 95\% CIs for
    observed (nontransported) and transported predicted effects. All
    results were approved for release by the U.S. Census Bureau,
    authorization numbers CBDRB-FY20-ERD002-023 and CBDRB-FY20-ERD002-024.}
\centering
\includegraphics[height=.8\textwidth,keepaspectratio]{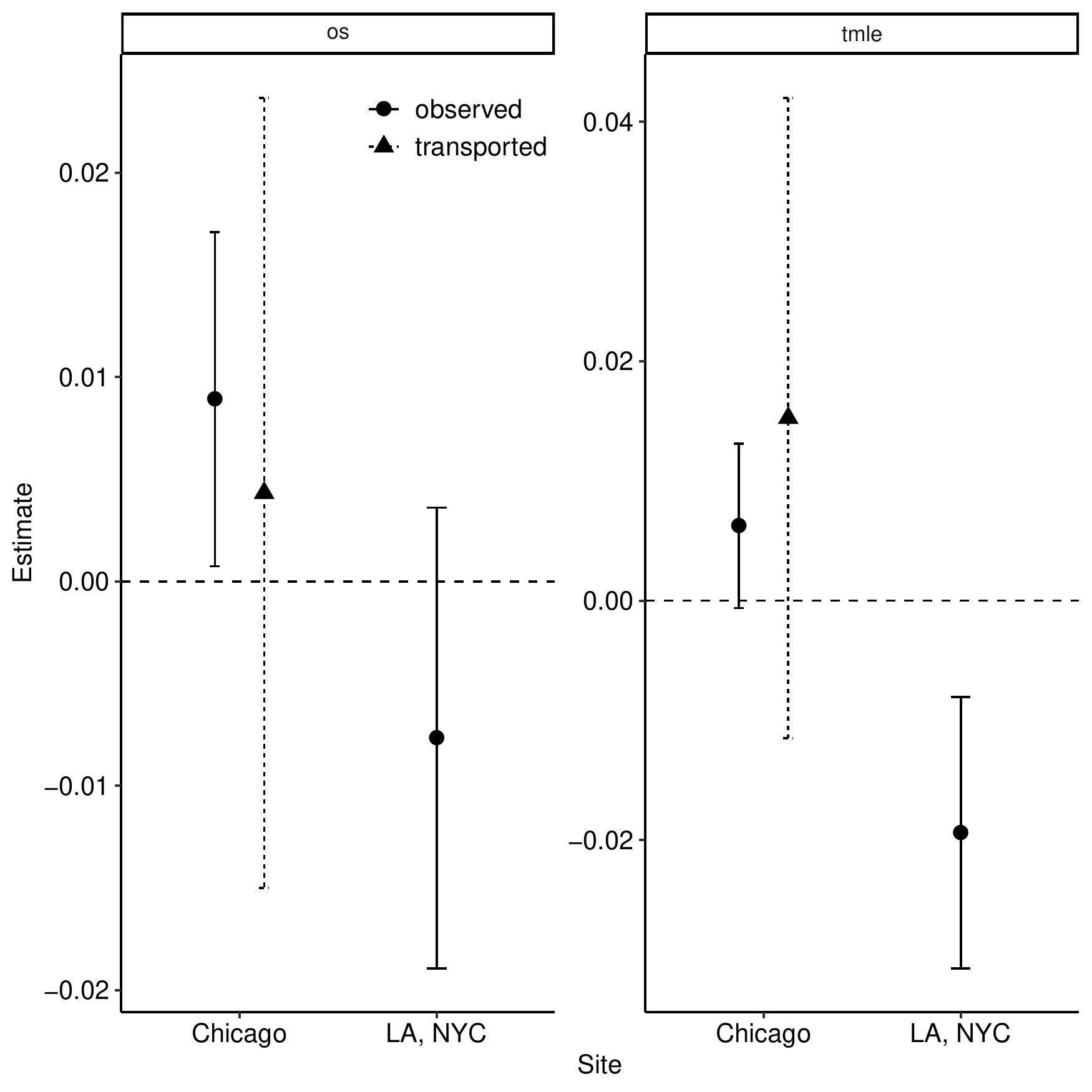}
\label{fig:examp}
\end{figure}

\section{Conclusions}
We proposed estimators for transported stochastic direct and indirect
effects under intermediate confounding and allowing for multiple,
possibly related mediating variables arising from a true, unknown
joint distribution. These estimators solve the efficient influence
function; one that does so in one step and the other that is a
substitution estimator that incorporates a series of targeting steps
to optimize the bias-variance trade-off. We derived their multiple
robustness properties and examined finite sample performance in a
simulation study. Lastly, we applied our proposed estimators to better
understand why a particular pathway from a housing intervention
through changes in the school environment resulted in an unintended
harmful effect on behavioral problems among girls in Chicago, when it
led to improvements in behavioral problems among girls in other
cities. 

\section*{Acknowledgements}
This research was conducted as a part of the U.S. Census Bureau's
Evidence Building Project Series. The U.S. Census Bureau has not
reviewed the paper for accuracy or reliability and does not endorse
its contents. Any conclusions expressed herein are those of the
authors and do not necessarily represent the views of the U.S. Census
Bureau. All results were approved for release by the U.S. Census
Bureau, authorization numbers CBDRB-FY20-ERD002-023 and CBDRB-FY20-ERD002-024.
\bibliographystyle{rss}
\bibliography{refs}

\newpage
\appendix

\section{Efficient influence function Theorem~1}
\begin{proof}
  In this proof we will use $\Theta(\P)$ to denote a parameter as a
  functional that maps the distribution $\P$ in the model to a real
  number. We will use $\Pn$ to denote the empirical distribution of $O_1, \ldots, O_n$. We will assume that the measure $\nu$ is discrete so that
  integrals can be written as sums. It can be checked algebraically
  that the resulting influence function will also correspond to the
  influence function of a general measure $\nu$. The true parameter
  value is thus given by
  \[\theta=\Theta(\P) = \sum_{y,z,m,w} y\,\p(y\mid a', z, m,
    w,1)\q(z\mid a',w,0)\p(m\mid a^{\star},w,0)\p(w\mid 0).\] The non-parametric MLE of
  $\theta$ is given by
  \begin{equation}
    \Theta(\Pn)=\sum_{y,z,m,w}y\frac{\Pn f_{y,a',z,m,w,1}}{\Pn
      f_{a',z,m,w,1}}\frac{\Pn f_{z,a',w,0}}{\Pn f_{a',w,0}}\frac{\Pn
      f_{m,a^{\star},w,0}}{\Pn f_{a^{\star},w,0}}\frac{\Pn
      f_{w,0}}{\Pn f_0}\label{nonpest},
  \end{equation}
  where we remind the reader of the notation $\P f =\int f \dd\P$. Here
  $f_{y,a,z,m,w,1}=\one(Y=y,A=a,Z=z,M=m,W=w,S=1)$, and $\one(\cdot)$ denotes
  the indicator function. The other functions $f$ are defined
  analogously.

  We will use the fact that the efficient influence function in a
  non-parametric model corresponds with the influence curve of the
  NPMLE. This is true because the influence curve of any regular
  estimator is also a gradient, and a non-parametric model has only
  one gradient. The Delta method shows that if $\hat \Theta(\Pn)$ is a
  substitution estimator such that $\theta=\hat \Theta(\P)$, and
  $\hat \Theta(\Pn)$ can be written as
  $\hat \Theta^{\star}(\Pn f:f\in\mathcal{F})$ for some class of functions
  $\mathcal{F}$ and some mapping $\Theta^{\star}$, the influence function of
  $\hat \Theta(\Pn)$ is equal to
  \[\dr_\P(O)=\sum_{f\in\mathcal{F}}\frac{\dd\hat
  \Theta^{\star}(\P)}{\dd\P f}\{f(O)-\P f\}.\]

  Applying this result to (\ref{nonpest}) with
  $\mathcal{F}=\{f_{y,a',z,m,w,1},f_{a',z,m,w,1},f_{z,a',w,0},f_{a',w,0},
  f_{m,a^{\star},w,0}, f_{a^{\star},w,0},f_{w,0},f_0:
  y,z,m,w\}$ and rearranging terms gives the result of the theorem.
  The algebraic derivations involved here are lengthy and not
  particularly illuminating, and are therefore omitted from the proof.
\end{proof}

\section{Lemma~1}
\begin{proof}
This result follows by replacing
\begin{equation}
  \p(m\mid a^{\star}, w)
  =\p(m\mid a', z, w)\frac{\g(a'\mid w)}{\g(a^{\star}\mid w)} \frac{\q(z\mid
    a',w)}{\rr(z\mid a',m,w)}\frac{\e(a^{\star}\mid m, w)}{\e(a'\mid
    m,w)}\label{eq:bayes}
\end{equation}
in expression (4) in the main text.
\end{proof}

\section{Additional results}
\begin{theorem}[Multiple robustness of the EIF]
  For notational simplicity, in this theorem we omit the dependence of all
  functions on $w$. Let $\P_W$ denote the distribution of $W$
  conditional on $S=0$. We have {\small
    \begin{align*}
      \P & D_{\eta_1} = \notag\\
         &\int
           \{\vv_1(a^{\star})-\vv(a^{\star})\}\left
           \{1-\frac{\g(a^{\star})}{\g_1(a^{\star})}\right\}
           \dd\P_W\dd\nu(m,z) + \\
         & \int \frac{\g(a')}{\g_1(a^{\star})}\frac{\bb(a',z,m)}{1-\bb(a',z,m)}\left
           \{\frac{1-\bb_1(a',z,m)}{\bb_1(a',z,m)}\frac{\q_1(z\mid a')}{\rr_1(z\mid
           a',m)}\frac{\e_1(a^{\star}\mid m)}{\e_1(a'\mid m)} -
           \frac{1-\bb(a',z,m)}{\bb(a',z,m)}\frac{\q(z\mid a')}{\rr(z\mid a',m)}\frac{\e(a^{\star}\mid m)}
           {\e(a'\mid m)}\right\}\times\\
      &\times\{\m(a',z,m) - \m_1(a',z,m)\}
           \p(m, z\mid a')\dd\P_W\dd\nu(m,z) + \\
         &\int\frac{\g(a')}{\g_1(a^{\star})}\left\{\uu(z,a') -
           \uu_1(z,a')\right\}\{\q_1(z\mid a')-\q(z\mid
           a')\}\dd\P_W\dd\nu(m,z)+\\
         &\int\frac{\g(a')}{\g_1(a^{\star})}\frac{\q(z\mid
           a')}{\rr(z\mid a',m)}\frac{\e(a^{\star} \mid m)}{\e(a'\mid
           m)}\left\{\m_1(a',z,m) - \m(a',z,m)\right\}\{\q_1(z\mid a')-\q(z\mid
           a')\}\p(m\mid z, a')\dd\P_W\dd\nu(m,z).
    \end{align*}}
\label{theo:mr}
\end{theorem}
\begin{proof}
  For fixed $a'$ and $a^{\star}$ we have
  \begin{align}
    \P & D_{\eta_1} = \notag\\
       & \int \frac{\g(a')}{\g_1(a^{\star})}\frac{1-\bb_1(a',z,m)}{\bb_1(a',z,m)}\frac{\bb(a',z,m)}{1-\bb(a',z,m)}\frac{\q_1(z\mid a')}{\rr_1(z\mid
         a',m)}\frac{\e_1(a^{\star}\mid m)}{\e_1(a'\mid m)}\times\notag\\
       &\times\{\m(a',z,m) -
         \m_1(a',z,m)\}\p(m\mid z, a')\q(z\mid a')
         \dd\P_W\dd\nu(m,z)\label{eq:t1}\\
    +&\int\frac{\g(a')}{\g_1(a^{\star})}\uu_1(z,a')\{\q(z\mid a') -
       \q_1(z\mid a')\}\dd\P_W\dd\nu(m,z)\label{eq:t2}\\
    +& \int \frac{\g(a^{\star})}{\g_1(a^{\star})}\m_1(a',z,m)
       \q_1(z\mid a')\p(m\mid a^{\star})\dd\P_W\dd\nu(m,z)\label{eq:t3}\\
    -&\int \vv_1(a^{\star})\left\{\frac{\g(a^{\star})}
       {\g_1(a^{\star})}-1\right\}\dd\P_W\dd\nu(m,z)\label{eq:t4}\\
    -&\int \vv(a^{\star})\dd\P_W\dd\nu(m,z).\label{eq:t5}
  \end{align}
  First, note that
  \begin{align}
    (\ref{eq:t4}) + (\ref{eq:t5}) &=\int \{\vv_1(a^{\star})-
    \vv(a^{\star})\}\left\{1-\frac{\g(a^{\star})}{\g_1(a^{\star})}\right\}
    \dd\P_W\dd\nu(m,z)\label{eq:t6}\\
    & - \int \vv(a^{\star})\frac{\g(a^{\star})}{\g_1(a^{\star})}
    \dd\P_W\dd\nu(m,z)\notag.
  \end{align}
  Using (\ref{eq:bayes}) above, we get
  \begin{align}
    \int\vv(a^{\star})&\frac{\g(a^{\star})}{\g_1(a^{\star})}\dd\P_W\\
            =\notag\\
            & \int\frac{\g(a')}{\g_1(a^{\star})}\frac{\q(z\mid a')}{\rr(z\mid
             a',m)}\frac{\e(a^{\star}\mid m)}{\e(a'\mid m)}\{\m(a',z,m)-
             \m_1(a',z,m)\}\p(m\mid z, a')\q(z\mid a')\dd\P_W\dd\nu(m,z)+
             \notag\\
            &\int\frac{\g(a')}{\g_1(a^{\star})}\frac{\q(z\mid a')}{\rr(z\mid
            a',m)}\frac{\e(a^{\star}\mid m)}{\e(a'\mid m)}\m_1(a',z,m)\p(m\mid
            z, a')\q(z\mid a')\dd\P_W\dd\nu(m,z)\label{eq:t9}
  \end{align}
  Thus
  \begin{align}
    (\ref{eq:t1})  + &(\ref{eq:t4}) + (\ref{eq:t5}) = (\ref{eq:t6})\notag\\
                                                   &+ \int \frac{\g(a')}{\g_1(a^{\star})}\left\{\frac{1-\bb_1(a',z,m)}{\bb_1(a',z,m)}\frac{\bb(a',z,m)}{1-\bb(a',z,m)}\frac{\q_1(z\mid a')}
      {\rr_1(z\mid a',m)}\frac{\e_1(a^{\star} \mid m)}{\e_1(a'\mid m)} -
      \frac{\q(z\mid a')}{\rr(z\mid a',m)}\frac{\e(a^{\star}\mid m)}{\e(a'\mid
      m)}\right\}\times\notag\\ &\,\,\,\,\,\,\,\,\{\m(a',z,m) -
      \m_1(a',z,m)\}\p(m\mid z, a')\q(z\mid a')
      \dd\P_W\dd\nu(m,z)\label{eq:t12}\\ &-(\ref{eq:t9})\notag.
  \end{align}
  Substituting (\ref{eq:bayes}) into
  (\ref{eq:t3}) we get
  \[(\ref{eq:t3})-(\ref{eq:t9})=\int\frac{\g(a')}{\g_1(a^{\star})}\frac{\q(z\mid
      a')}{\rr(z\mid a',m)}\frac{\e(a^{\star}\mid m)}{\e(a'\mid
      m)}\m_1(a',z,m)\p(m\mid z, a')\{\q_1(z\mid a')-\q(z\mid
    a')\}\dd\P_W\dd\nu(m,z),\]
  which yields
  {\small
  \begin{align}
    (\ref{eq:t2})&+(\ref{eq:t3})-(\ref{eq:t9})\notag\\
                 &=\int\frac{\g(a')}{\g_1(a^{\star})}\left\{\frac{\q(z\mid
                   a')}{\rr(z\mid a',m)}\frac{\e(a^{\star}\mid m)}{\e(a'\mid
                   m)}\m_1(a',z,m)\p(m\mid z, a') -
                   \uu_1(z,a')\right\}\{\q_1(z\mid a')-\q(z\mid
                   a')\}\dd\P_W\dd\nu(m,z)\notag\\
                   &=\int\frac{\g(a')}{\g_1(a^{\star})}\left\{\uu(z,a') -
                  \uu_1(z,a')\right\}\{\q_1(z\mid a')-\q(z\mid
                   a')\}\dd\P_W\dd\nu(m,z)\label{eq:t10}\\
                  &+\int\frac{\g(a')}{\g_1(a^{\star})}\frac{\q(z\mid
                   a')}{\rr(z\mid a',m)}\frac{\e(a^{\star}\mid m)}{\e(a'\mid
                   m)}\left\{\m_1(a',z,m) - \m(a',z,m)\right\}\{\q_1(z\mid a')-\q(z\mid
                   a')\}\p(m\mid z, a')\dd\P_W\dd\nu(m,z)\label{eq:t11}
  \end{align}}
  Putting everything together yields
  \[(\ref{eq:t1}) +(\ref{eq:t2}) + (\ref{eq:t3}) + (\ref{eq:t4}) +
    (\ref{eq:t5})=(\ref{eq:t6}) +(\ref{eq:t12})+(\ref{eq:t10})
    +(\ref{eq:t11}),\]
  yielding the result of the theorem.
\end{proof}

\end{document}